\documentclass{article}
\pdfoutput=1
\usepackage[utf8]{inputenc}
\usepackage{subfig,ifthen,graphicx,enumerate}
\usepackage{amsmath,amsthm,amssymb}
\usepackage{xspace,paralist}
\usepackage{calc}
\usepackage{aliascnt,hyperref}
\usepackage{cleveref}
\usepackage[round]{natbib}
\let\cite\citep

\ifx\pgfversion\undefined
  \input{pgfexternal.tex}
\else

  \usetikzlibrary{shapes.geometric}
  \usetikzlibrary{shapes.callouts}
  \usetikzlibrary{calc,fit}
  \pgfdeclarelayer{background}
  \pgfsetlayers{background,main}

  \tikzstyle{commgraphnode}=[circle,minimum size=20pt]
  \tikzstyle{commgraphedge}=[]
  \tikzstyle{commgraphhyperarc}=[ellipse,fill=gray,opacity=.3,inner sep=0pt]

  \ifthenelse{\lengthtest{\pgfversion cm<1.18cm}}{
  }{
    \ifx\regeneratepgf\undefined\pgfrealjobname{\jobname}
    \else\pgfrealjobname{\regeneratepgf}
    \fi
  }
\fi

\title{Common Knowledge in Interaction Structures}
\author{Krzysztof R. Apt and Andreas Witzel and Jonathan A. Zvesper} 

\ifx\pdfoutput\undefined\else\fi

\ifx\crefformat\undefined
\newcommand{\crefformats}[7]{}
\else
\newcommand{\crefformats}[7]{%
  \crefformat{#3}{##2\ifthenelse{\equal{#4}{}}{}{#4~}#1##1#2##3}
  \Crefformat{#3}{##2\ifthenelse{\equal{#5}{}}{}{#5~}#1##1#2##3}
  \crefrangeformat{#3}{\ifthenelse{\equal{#6}{}}{}{#6~}##3#1##1#2##4--##5#1##2#2##6}
  \Crefrangeformat{#3}{\ifthenelse{\equal{#7}{}}{}{#7~}##3#1##1#2##4--##5#1##2#2##6}
  \crefmultiformat{#3}{\ifthenelse{\equal{#6}{}}{}{#6~}##2#1##1#2##3}{ and~##2#1##1#2##3}{, ##2#1##1#2##3}{ and~##2#1##1#2##3}
  \Crefmultiformat{#3}{\ifthenelse{\equal{#7}{}}{}{#7~}##2#1##1#2##3}{ and~##2#1##1#2##3}{, ##2#1##1#2##3}{ and~##2#1##1#2##3}
  \crefrangemultiformat{#3}{\ifthenelse{\equal{#6}{}}{}{#6~}##2#1##1#2##3}{ and~##2#1##1#2##3}{, ##2#1##1#2##3}{ and~##2#1##1#2##3}
  \Crefrangemultiformat{#3}{\ifthenelse{\equal{#7}{}}{}{#7~}##2#1##1#2##3}{ and~##2#1##1#2##3}{, ##2#1##1#2##3}{ and~##2#1##1#2##3}
}

\makeatletter
\@ifclassloaded{llncs}{
  \crefformats{}{}{section}{Sect.}{Section}{Sects.}{Sections}
  \crefformats{}{}{appendix}{App.}{Appendix}{App.}{Appendices}
  \crefformats{}{}{subsection}{Sect.}{Section}{Sects.}{Sections}
  \crefformats{}{}{subsubsection}{Sect.}{Section}{Sects.}{Sections}
  \crefformats{}{}{chapter}{Chap.}{Chapter}{Chaps.}{Chapters}
  \crefformats{}{}{part}{Part}{Part}{Parts}{Parts}
  \crefformats{}{}{figure}{Fig.}{Figure}{Figs.}{Figures}
  \crefformats{}{}{table}{Table}{Table}{Tables}{Tables}
}{
   \crefformats{}{}{section}{Section}{Section}{Sections}{Sections}
   \crefformats{}{}{appendix}{Appendix}{Appendix}{Appendices}{Appendices}
   \crefformats{}{}{subsection}{Section}{Section}{Sections}{Sections}
   \crefformats{}{}{subsubsection}{Section}{Section}{Sections}{Sections}
   \crefformats{}{}{chapter}{Chapter}{Chapter}{Chapters}{Chapters}
   \crefformats{}{}{part}{Part}{Part}{Parts}{Parts}
   \crefformats{}{}{figure}{Figure}{Figure}{Figures}{Figures}
   \crefformats{}{}{table}{Table}{Table}{Tables}{Tables}
}
\makeatother
\fi

\ifx\finalversion\undefined
\newcommand{\marginlabel}[2]{%
  \mbox{}%
  \marginpar[\raggedleft\hspace{0pt}#1]{\raggedright\hspace{0pt}#2}%
}
\ifx\tikzstyle\undefined
\newcommand{\todoar}[2][]{\todo[#1]{#2}}
\else\tikzstyle{todoarrow}=[opacity=0.4,gray,-stealth]
\newcommand{\todoar}[2][]{%
  \marginlabel{\small #2}
        {\tikz[remember picture,overlay,baseline=(todoarrowstart.220)]\node(todoarrowstart){};$\lhd$ \small #2}%
  \ifthenelse{\equal{#1}{}}{}{{\color{red}[}#1{\color{red}]}}%
  \tikz[remember picture,overlay]\node[inner sep=2pt](todoarrowend){};%
  \tikz[remember picture,overlay]\path(todoarrowstart)edge[todoarrow,out=190,in=-45](todoarrowend);%
}
\fi
\newcommand{\todo}[2][]{%
  \marginlabel{{\small #2} $\rhd$}{$\lhd$ \small #2}%
  \ifx\color\undefined%
  \ifthenelse{\equal{#1}{}}{}{{[}#1{]}}%
  \else%
  \ifthenelse{\equal{#1}{}}{}{{\color{red}[}#1{\color{red}]}}%
  \fi%
}
\fi

\newcounter{autoexternalpgf}
\newcommand{\abs}[1]{\lvert#1\rvert}

\newenvironment{mainclaim}{\begin{center}}{\end{center}}

\makeatletter
\@ifclassloaded{beamer}{}{%
  \newcommand{\newtheoremwithalias}[3]{%
    \ifx\newaliascnt\undefined
    \newcounter{#1}
    \else
    \newaliascnt{#1}{#2}
    \fi
    \newtheorem{#1}[#1]{#3}
    \ifx\aliascntresetthe\undefined\else
    \aliascntresetthe{#1}
    \fi
  }
  \ifx\creflastconjunction\undefined\else
  \renewcommand{\creflastconjunction}{ and }
  \fi

  \@ifclassloaded{llncs}{
    \if@envcntsame\errmessage{cleveref naming doesn't work because no aliascntrs used in llncs.cls}
    \else\if@envcntsect
    \newtheorem{observation}{Observation}[section]
    \newtheorem{fact}{Fact}[section]
    \else
    
    \newtheorem{fact}{Fact}
    \fi\fi
    
    \crefformats{(}{)}{equation}{}{Equation}{}{Equations}
  }{
    \newtheorem{theorem}{Theorem}[section]
    \newtheoremwithalias{lemma}{theorem}{Lemma}
    \newtheoremwithalias{proposition}{theorem}{Proposition}
    \newtheoremwithalias{corollary}{theorem}{Corollary}
    \newtheoremwithalias{observation}{theorem}{Observation}
    \newtheoremwithalias{fact}{theorem}{Fact}
    \newtheoremwithalias{conjecture}{theorem}{Conjecture}
    \newtheoremwithalias{numberednote}{theorem}{Note}
    \newtheoremwithalias{numberedremark}{theorem}{Remark}

    \ifx\theoremstyle\undefined\else\theoremstyle{remark}\fi

    \ifx\theoremstyle\undefined\else\theoremstyle{definition}\fi
    \newtheoremwithalias{definition}{theorem}{Definition}
    \newtheoremwithalias{example}{theorem}{Example}

    \crefformats{(}{)}{equation}{Equation}{Equation}{Equations}{Equations}

  }

  \crefformats{}{}{theorem}{Theorem}{Theorem}{Theorems}{Theorems}
  \crefformats{}{}{lemma}{Lemma}{Lemma}{Lemmas}{Lemmas}
  \crefformats{}{}{proposition}{Proposition}{Proposition}{Propositions}{Propositions}
  \crefformats{}{}{corollary}{Corollary}{Corollary}{Corollaries}{Corollaries}
  \crefformats{}{}{observation}{Observation}{Observation}{Observations}{Observations}
  \crefformats{}{}{fact}{Fact}{Fact}{Facts}{Facts}
  \crefformats{}{}{conjecture}{Conjecture}{Conjecture}{Conjectures}{Conjectures}
  \crefformats{}{}{numberednote}{Note}{Note}{Notes}{Notes}
  \crefformats{}{}{numberedremark}{Remark}{Remark}{Remarks}{Remarks}
  \crefformats{}{}{definition}{Definition}{Definition}{Definitions}{Definitions}
  \crefformats{}{}{example}{Example}{Example}{Examples}{Examples}
}
\makeatother

\newcommand{\oldbfe}[1]{\begin{bfseries}\emph{#1}\end{bfseries}}

\newcommand{\ES}{\mbox{$\emptyset$}}

\newcommand{\La}{\mbox{$\:\Leftarrow\:$}}
\newcommand{\Ra}{\mbox{$\:\Rightarrow\:$}}

\newcommand{\sse}{\mbox{$\:\subseteq\:$}}

\newcommand{\LL}{\mbox{$\ldots$}}

\newcommand{\C}[1]{\mbox{$\{{#1}\}$}}           

\newcommand{\NI}{\noindent}

\newcommand{\II}{\vspace{2 mm}}

\newcommand{\szkew}[1]{\relax \setbox0=\hbox{\kern -24pt $\displaystyle#1$\kern 0pt }%
\box0}
{\catcode`\@=11 \global\let\ifjusthvtest@=\iffalse}

\newcounter{oldmycaption}

\title{Common Knowledge in Interaction Structures\footnote{To appear in Proceedings of TARK 2009}}
\author{Krzysztof R. Apt
\and
Andreas Witzel
\and
Jonathan A. Zvesper
\\[2ex]
ILLC, Universiteit van Amsterdam, Netherlands \\
and CWI, Amsterdam, Netherlands
}

\newcommand{\oldendexample}{}
\let\oldendexample\endexample
\renewcommand{\endexample}{\qed\oldendexample}

\newcommand{\state}[1][]{\ensuremath{(V#1,M#1)}\xspace}
\newcommand{\dfn}[1]{\emph{\bfseries #1}}
\newcommand{\setof}[1]{\ensuremath{\mathit{Set}(#1)}\xspace}
\newcommand{\msg}[3]{\ensuremath{(#1,#2,#3)}\xspace}
\newcommand{\bits}{\ensuremath{At}\xspace}
\newcommand{\knows}[1]{\ensuremath{K_{#1}}\xspace}
\newcommand{\ck}[1]{\ensuremath{C_{#1}}\xspace}
\newcommand{\Facts}{\mathit{Facts}}

\newcommand{\powerset}[1]{\ensuremath{\mathcal{P}(#1)}}

\renewcommand{\iff}{\text{ iff }}
\renewcommand{\enspace}{}

\begin{document}

\maketitle
\sloppy

\begin{abstract}
  We consider two simple variants of a framework
  for reasoning about knowledge amongst communicating groups of players.
  Our goal is to clarify the resulting epistemic issues.
  In particular, we investigate what is the impact of common knowledge of the
  underlying hypergraph connecting the players, and under what conditions
  common knowledge distributes over disjunction.
  We also obtain two versions of the classic result that common knowledge cannot be
  achieved in the absence of a simultaneous event
  (here a message sent to the whole group).
\end{abstract}

\section{Introduction}
\label{sec:intro}

We introduce a framework for reasoning about communication amongst groups of players.
We assume that each player is a member of a certain number of groups, and that he is able to broadcast
synchronously information to each of those groups.
Thus there is what we call an \dfn{interaction structure}, a hypergraph of the players, that determines the
communication protocol.
We are interested in studying what players can learn in certain restricted communication settings,
what impact common knowledge of the underlying hypergraph can have,
and in properties of the resulting knowledge that can simplify reasoning about it.

\begin{figure}
  \centering
  \subfloat[]{\label{fig:interaction-structure-examples:a}
    \beginpgfgraphicnamed{graph_kite1}
    \begin{tikzpicture}[commgraphnode,rounded corners=5mm,scale=.7,transform shape]
      \node[draw] (m) at (1,0)  {$n$};
      \node[draw] (l) at (2,1)  {$l$};
      \node[draw] (k) at (2,-1) {$k$};
      \node[draw] (j) at (4,0)  {$j$};
      \node[draw] (i) at (5.5,0)  {$i$};

      \begin{pgfonlayer}{background}
      \fill[commgraphhyperarc,overlay] ($(m.west)+(-18pt,0)$) -- ($(k.south east)+(6pt,-18pt)$) -- ($(l.north east)+(6pt,18pt)$) -- cycle;
      \path
        node at (1.9,.5) (ljl) {}
        node at (4.1,.5) (ljr) {}
        (l) -- node[sloped,commgraphhyperarc,fit=(ljl)(ljr)] {} (j);
      \path
        node at (1.9,-.5) (kjl) {}
        node at (4.1,-.5) (kjr) {}
        (k) -- node[sloped,commgraphhyperarc,fit=(kjl)(kjr)] {} (j);
      \node[commgraphhyperarc,fit=(j)(i)] {};
      \end{pgfonlayer}
    \end{tikzpicture}
    \endpgfgraphicnamed
  }
  \subfloat[]{\label{fig:interaction-structure-examples:b}
    \beginpgfgraphicnamed{graph_kite2}
    \begin{tikzpicture}[commgraphnode,scale=.7,transform shape]
      \node[draw] (m) at (1,0)  {$n$};
      \node[draw] (l) at (2,1)  {$l$};
      \node[draw] (k) at (2,-1) {$k$};
      \node[draw] (j) at (3,0)  {$j$};
      \node[draw] (i) at (4.5,0)  {$i$};

      \begin{pgfonlayer}{background}
      \path
        node at (.9,.5) (mll) {}
        node at (2.1,.5) (mlr) {}
        (m) -- node[sloped,commgraphhyperarc,fit=(mll)(mlr)] {} (l);
      \path
        node at (.9,-.5) (mkl) {}
        node at (2.1,-.5) (mkr) {}
        (m) -- node[sloped,commgraphhyperarc,fit=(mkl)(mkr)] {} (k);
      \path
        node at (1.9,.5) (ljl) {}
        node at (3.1,.5) (ljr) {}
        (l) -- node[sloped,commgraphhyperarc,fit=(ljl)(ljr)] {} (j);
      \path
        node at (1.9,-.5) (kjl) {}
        node at (3.1,-.5) (kjr) {}
        (k) -- node[sloped,commgraphhyperarc,fit=(kjl)(kjr)] {} (j);
      \node[commgraphhyperarc,fit=(j)(i)] {};
      \end{pgfonlayer}
    \end{tikzpicture}
    \endpgfgraphicnamed
  }
  \caption{Two interaction structures. Hyperarcs are shown in gray.}
  \label{fig:interaction-structure-examples}
\end{figure}
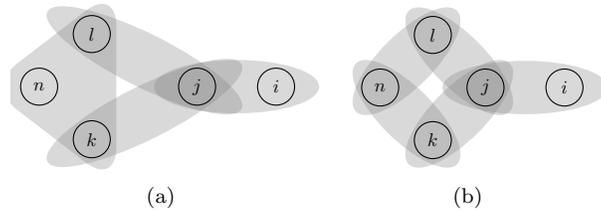

For example, consider \cref{fig:interaction-structure-examples}.
If player~$i$ knows that he is in interaction structure~\subref{fig:interaction-structure-examples:a},
and he learns a fact from player~$j$ that initially only player~$n$ knew,
then $i$~can deduce that both~$l$ and~$k$ also must have learned that fact.
In interaction structure~\subref{fig:interaction-structure-examples:b},
he can only deduce that either of them has learned it, but not which one.
If~$i$ does not know the interaction structure,
he cannot draw such conclusions, since player~$n$ might as well have communicated with player~$j$ directly.
One particular focus of our discussion concerns conditions under which knowledge of a disjunction
does allow us to deduce knowledge of one particular disjunct,
thus simplifying reasoning in such situations.
Another focus is to analyze the conditions for attaining common knowledge.

In the following \cref{sec:defs}, we first set up a more restricted framework
where players can only communicate those facts that they initially know,
and we examine this framework in detail in \cref{sec:telling}.
In \cref{sec:forwarding} we then lift this restriction
and examine how the properties of knowledge are affected
when players are allowed to send information which they learned from other players.
In \cref{sec:conclusions} we discuss related work,
in particular two closely related frameworks from the literature,
and draw some conclusions.
We look at some possible extensions in \cref{sec:extensions}.

\section{Preliminaries}
\label{sec:defs}

We assume the following setup to be common knowledge among the players.
There is a set of players $N$.  Each player $i \in N$ has
a private set $\bits_i$ of facts (atomic propositions), of which only player $i$ initially knows whether they are true.
The truth values of these facts are represented by a \dfn{valuation},
which can be written as a set $V \subseteq \bits$ containing those facts that are true,
where $\bits=\bigcup_{i\in N}\bits_i$.
By $V_i$, we denote $V\cap \bits_i$, the restriction of $V$ to $i$'s facts. 


Throughout, we assume communication to be \dfn{truthful} in the sense that it only contains information the sender knows to be true.

An \dfn{interaction structure} for players $N$ is a tuple $(H,
(\bits_i)_{i \in N})$, where $H$ is a \dfn{hypergraph on $N$}, i.e.,  a
set of non-empty subsets of $N$, called \dfn{hyperarcs}, and the $At_i$ are pairwise disjoint sets.

In the present section we place two restrictions, that are related.
Firstly, we use unordered \emph{sets} of messages,
i.e.~without any temporal structure, since it only matters
whether a given message has been broadcast or not,
and not \emph{when} it was broadcast.  Secondly, we only allow
messages of the form $\msg{i}{A}{p}$ with $i\in A\in H$ and $p\in\bits_i$.
That is, players only broadcast basic facts that `belong' to them.
In \cref{sec:forwarding} we partially lift these restrictions,
allowing more general forms of broadcast.  This in turn
means introducing some temporal ordering since if the
message $\msg{i}{A}{p}$ occurs, with $p\not\in\bits_i$, then everybody in
$A$ knows that \emph{before} that broadcast there was another
broadcast of the form $\msg{\cdot}{B}{p}$ with $i \in B$,
since otherwise~$i$ could not have known~$p$.

Given these restrictions, we consider two different situations:
one in which the underlying hypergraph is commonly known amongst the players;
and one in which it is not,
in the sense that a player knows only the hyperarcs to which he belongs.

In each case an interaction structure defines a
communication protocol: each player~$i$ can at any point broadcast
any true fact $p \in \bits_i$ to any hyperarc $A \in H$ with $i \in A$.
Thus a \dfn{message} is a tuple $\msg{i}{A}{p}$
with $i \in A$ and $p\in\bits_i$; $\msg{i}{A}{p}$ is the message in which $i$
communicates among the group $A$ his fact~$p$.
\dfn{$H$-compliant messages} are those in which $A \in H$.
If the players consider only $H$-compliant messages possible,
then they know the underlying hypergraph $H$.
So if the model allows only $H$-compliant messages, the
underlying hypergraph $H$ is common knowledge among the players;
if it uses all messages, $H$ is unknown.

We next define our model formally in order to reason about the
knowledge of the players and how it changes as messages are broadcast.
This is roughly along the lines of \emph{history based models}
(see, e.g., \citet{pacuit_reasoning_2007,FHMV_RAK}).
We start by defining a \dfn{state}, which we might also have called `possible world', \state
to consist of a valuation~$V \subseteq \bits$
and a set~$M$ of messages $\msg{\cdot}{\cdot}{p}$ such that~$p \in V$.
An \dfn{$H$-compliant} state is one where $M$ only contains $H$-compliant messages.

A \dfn{word} over a set $A\subseteq N$ is a finite sequence $w=i_1\ldots i_k$ where each $i_l \in A$.
By $A^*$ we denote the set of all words over $A$,
and we write $\setof w$ for the \emph{set} of players occurring in $w$.

Now given a set of messages $M$ and a word $w$, we introduce the following notation: 
\begin{align*}
M_w &:= \C{(\cdot,A,\cdot) \in M \mid \setof w\subseteq A}\\
\Facts(M) &:= \C{p \mid (\cdot, \cdot,p) \in M}.
\end{align*}
So $M_i$ (respectively, $M_w$) is the subset of the set of messages~$M$ that player $i$ received
(respectively, that were broadcast to all the players in~$w$; note that the order in $w$ does not matter),
and $\Facts(M)$ is the set of facts that were communicated in the messages in~$M$.
In particular, $\Facts(M_i)$ is the set of facts that were communicated
in the messages in~$M$ that player~$i$ received.
Note that \state is a state if $\Facts(M)\subseteq V$.
Further, we define all set operations to act component-wise on states, e.g.
$\state \subseteq \state[']$ iff $V \subseteq V'$ and $M \subseteq M'$.

In order to represent the knowledge of the players we define
an \dfn{indistinguishability relation} between states:
$ 
\mbox{$\state  \sim_i \state[']$ iff $(V_i, M_i) = (V'_i, M'_i)$.}
$ 

In the semantics we present below, a player $i$ is said to `know' a fact
just if that fact is true in every state that is indistinguishable for $i$
from the actual state.  Of particular interest to us is the knowledge of
\emph{groups} $G\subseteq N$ (always assumed to be non-empty).
Specifically we consider the so-called
`common knowledge' among a group (cf.~\cite[p.~23]{FHMV_RAK}).
These are facts that everybody in the group knows,
they all know that they know, etc.
To define this formally we extend the individual
indistiguishability relation to groups:
for $G \subseteq N$ the relation $\sim_G$ is the transitive
closure of $\bigcup_{i \in G}\sim_i$.

We are interested in properties definable by the following  \dfn{epistemic language} $\mathcal{L}$:
\[
 \varphi ::= p \mid \neg \varphi \mid \varphi \land \varphi \mid \varphi \lor \varphi \mid \ck G \varphi,
\]
where the atoms $p$ denote the facts in~$\bits$,
$\neg$, $\land$ and $\lor$ are the standard connectives;
and $\ck G$ is a knowledge operator, with $\ck G \varphi$ meaning $\varphi$
is common knowledge among $G$.
We write $\knows i$ for $\ck{\{i\}}$;
$\knows i \varphi$ can be read `$i$ knows that $\varphi$'.
The \dfn{positive language} $\mathcal{L}^+$ is the sublanguage of $\mathcal{L}$
in which negation ($\neg$) does not occur.

The semantics for $\mathcal{L}$ is as follows:
\begin{align*}
  \state  &\vDash_H p                   && \iff p \in V, \\
  \state  &\vDash_H \neg \varphi            && \iff \state  \nvDash_H \varphi, \\
  \state  &\vDash_H \varphi \lor \psi       && \iff \state  \vDash_H \varphi\text{ or }\state  \vDash_H \psi, \\
  \state  &\vDash_H \varphi \land \psi      && \iff \state  \vDash_H \varphi\text{ and }\state  \vDash_H \psi, \\
  \state  &\vDash_H \ck G \varphi             && \iff
  \begin{aligned}[t]
    &\state['] \vDash_H \varphi\\
    &\text{for each $H$-compliant \state[']}\\
    &\text{with } \state \sim_{G} \state[']\enspace.
  \end{aligned}
\end{align*}

By allowing only $H$-compliant states in the last clause of the semantics,
the underlying hypergraph~$H$ is assumed to be \emph{common knowledge}.
Assuming that the hypergraph $H$ is \emph{unknown} turns out to be equivalent to 
the case where it is common knowledge that the hypergraph $H$ is complete,
i.e., $H=\powerset{N} - {\emptyset}$.
This might seem counter-intuitive, but it reflects the fact that if the
hypergraph is unknown then every player must consider it possible that
every set $A \subseteq N$ might be a hyperarc in $H$.
To denote the corresponding semantics,
we use $\vDash$ as abbreviation for  $\vDash_H$ with $H$ being the complete hypergraph.

For a word $w=i_1\ldots i_k$, we write $\knows w$ to abbreviate $\knows{i_1}\knows{i_2} \ldots \knows{i_k}$,
and write $\sim_w$ to denote the concatenation $\sim_{i_1} \circ \ldots \circ \sim_{i_k}$.

Notice that $\state \sim_G \state[']$ iff there is $w\in G^*$ with $\state \sim_w \state[']$.
So an equivalent way of specifying the semantics for $\ck G$ with non-singleton $G$ is as follows:
\begin{equation}
  \label{equ:K}
  \state \vDash \ck G \varphi\text{ iff }\state \vDash \knows w \varphi\text{ for all }w\in G^*\enspace.
  \tag{$\star$}
\end{equation}
We now study the consequences of two choices in the analysis of players' knowledge:
\begin{itemize}
\item The type of messages; we assumed already that players send only atomic information, but there still remains a choice whether, as assumed above, players only send information they know initially, or can can send information that they have learned from other players.  The former scenario is explored in \cref{sec:telling}, the latter in \cref{sec:forwarding}.
\item The issue whether the underlying hypergraph is commonly known among the players.
We consider this distinction in both of the following sections.
\end{itemize}

We shall see that both choices have bearing on players' knowledge.

\section{Telling}
\label{sec:telling}

In this section we study the case under the assumption mentioned above,
that players' messages refer only to the facts they know initially.
So players can send only information they know at the outset.
We call this contingency `telling'.

For the relevance of common knowledge of~$H$, consider the following example.
\begin{example}
  \label{ex:ck-of-h-matters-for-negative}
  For players $G=\{i,j,k\}$, $H=\{G\}$, $p\in \bits_k$ and $\state=(\emptyset,\emptyset)$,
  we have
  \[
  \state\vDash_H\knows i\neg \knows j p\enspace.
  \]
  Indeed, the only hyperarc in $H$ through which player~$j$ could learn anything from $k$
  is the one which also contains player~$i$.
  So there is no way for $k$ to tell $j$ anything `secretly'.
  Hence, with $M_i=\emptyset$,
  $i$ also knows that $M_{jk}=\emptyset$.
  That is, in all states which $i$ considers possible at \state,
  $k$ has not told $j$ that $p$,
  therefore in all these states $j$ does not know $p$.

  On the other hand, we have
$ 
  \state\nvDash \knows i\neg \knows j p\enspace,
$ 
  since $\state\sim_i(\{p\},\{\msg{k}{\{j,k\}}{p}\})$.

  So for some formulas, common knowledge of $H$ matters.
  Note also that in this example, we even have
  \[
  \state\vDash_H\ck G\neg \knows j p\enspace,
  \]
  which shows that common knowledge can be attained without any communication taking place.
\end{example}

However, common knowledge of formulas from the positive language $\mathcal{L}^+$
can only be attained through messages received by the whole group,
and for these formulas, common knowledge of~$H$ does not matter.
In order to establish this,
we first show the following \cref{result:positive-keep-holding}, which intuitively says that
if a positive formula is true in some state,
then it remains true in any state where more facts are true
or more communication has taken place.
Remember that $\vDash$ corresponds to $\vDash_H$ with $H$ being the complete hypergraph,
so the following carries over to general states and $\vDash$.
\begin{lemma}
  \label{result:positive-keep-holding}
  For any $\varphi\in\mathcal{L}^+$ and $H$-compliant states $\state $ and \state[']
  with $\state\subseteq\state[']$,
\[
    \textup{if } \state\vDash_H\varphi, \textup{ then } \state['] \vDash_H\varphi.
\]
\end{lemma}
\begin{proof}
  We proceed by structural induction on $\varphi$.
  The only not completely obvious case is when
  $\varphi=\ck G\psi$ with $\psi\in\mathcal{L}^+$.
  We show the claim for $G=\{i\}$; the non-singleton case then follows by induction and~\eqref{equ:K}.
  Take an $H$-compliant state \state[''] such that $\state[']\sim_i\state['']$.
  Let 
  \[
  \state['''] := (V_i\cup\textstyle\bigcup_{j\neq i}V_j'',M_i)\enspace.
  \]
  We have $\Facts(M_i)\subseteq\Facts(M)\subseteq V$,
  since \state is a state.
  Also, $M_i\subseteq M_i'=M_i''\subseteq M''$,
  so $\Facts(M_i)\subseteq\Facts(M'')\subseteq V''$,
  since \state[''] is a state.
  Hence,
$ 
  \Facts(M''')=\Facts(M_i)\subseteq V\cap V''=\textstyle\bigcup_{i\in N}(V_i\cap V_i'')\subseteq V'''\enspace.
$ 
  This shows that \state['''] is a state. 
  Moreover, $\state\sim_i\state[''']$.
  Assume now $\state\vDash_H \knows i\psi$.
  Then we obtain $\state[''']\vDash_H\psi$.
  Further, we have $\state[''']\subseteq\state['']$
  since $V_i\subseteq V_i'=V_i''$ and $M_i\subseteq M_i'=M_i''$ due to $\state[']\sim_i\state['']$.
  Thus, by induction hypothesis we obtain $\state['']\vDash_H\psi$.
\end{proof}

\begin{theorem}
  \label{result:ck-of-h-doesnt-matter}
  For any $H$-compliant state $\state$ and $\varphi\in\mathcal{L}^+$,
\[
\mbox{$\state\vDash\varphi$ iff $\state\vDash_H\varphi$.}
\]
\end{theorem}
\begin{proof}
  We proceed by structural induction.
  The only non-trivial step is when $\varphi=\ck G\psi$ with $\psi\in\mathcal{L}^+$.
  
  \noindent ($\Rightarrow$) By induction hypothesis, $\state\vDash\ck G\psi$ implies $\state\vDash_H\ck G\psi$,
  since each $H$-compliant state is also a state.

  \noindent ($\Leftarrow$) Assume to the contrary that
  $\state\nvDash\ck G\psi$.
  So there is
  a state $\state[']$ with $\state\sim_G \state[']$ and $ \state[']\nvDash\psi$.
  Now let
  \[
  M'\restriction_H:=\{\msg{\cdot}{A}{\cdot} \in M' \mid A\in H\}\enspace.
  \]
  So $M'\restriction_H$ consists of all $H$-compliant messages of $M'$.
  Now note that $( V', M'\restriction_H)\subseteq \state[']$,
  so from $ \state[']\nvDash\psi$
  we obtain that $(V', M'\restriction_H)\nvDash\psi$
  using \cref{result:positive-keep-holding} (which, as noted, also holds for general states and $\vDash$).
  Since $( V', M'\restriction_H)$ is $H$-compliant, 
  the induction hypothesis yields $(V', M'\restriction_H)\nvDash_H\psi$.
  Moreover, we also have $\state\sim_G(V', M'\restriction_H)$,
  since \state is $H$-compliant and $\state\sim_G \state[']$.
  Thus, $\state\nvDash_H\ck G\psi$.
\end{proof}

In the remainder of this section, we are concerned with formulas from~$\mathcal{L}^+$,
so in view of the above results we restrict attention to $\vDash$.

We now establish that $\ck G$ distributes over disjunctions of positive formulas,
starting with singleton~$G$.
\begin{lemma}
  \label{result:k-disjunction-distributes}
  For any $\varphi_1,\varphi_2\in\mathcal{L}^+$, $i\in N$, and state \state,
\[
\mbox{$\state\vDash \knows i(\varphi_1\vee\varphi_2)$ iff $\state\vDash \knows i\varphi_1\vee \knows i\varphi_2$.}
\]
\end{lemma}
\begin{proof}
  To deal with the ($\Rightarrow$) implication assume that
  $\state\nvDash \knows i\varphi_1\vee \knows i\varphi_2$.  Then
  $\state\nvDash \knows i\varphi_1$ and $\state\nvDash \knows i\varphi_2$,
  i.e., there are $ \state[']$ and $\state['']$ such that
  \begin{align*}
    \state\sim_i \state[']&\text{ and }\state[']\nvDash\varphi_1\enspace\text{, as well as}\\
    \state\sim_i\state['']&\text{ and }\state['']\nvDash\varphi_2\enspace.
  \end{align*}
  Let now 
  \[
  \state[''']:=( V_i\cup\textstyle\bigcup_{j\neq i}(V_j'\cap V_j''),M_i)\enspace.
  \]

  Then $\Facts(M)\subseteq V$, $\Facts(M')\subseteq V'$, $\Facts(M'')\subseteq V''$,
  since \state, \state['] and \state[''] are states.
  Moreover, $M_i=M_i'$ and $M_i=M_i''$.
  So,
  \begin{align*}
    \Facts(M''')&\subseteq\Facts(M)\cap\Facts(M')\cap\Facts(M'')\\
    &\subseteq V\cap V'\cap V''\\
    &=\textstyle\bigcup_{i\in N}(V_i\cap V_i'\cap V_i'')\\
    &\subseteq V'''\enspace.
  \end{align*}
  This shows that \state['''] is a state,
  and since $M'''=M_i\subseteq M$ it is $H$-compliant.

  Now since $V_i=V_i'=V_i''$ and $M_i=M_i'=M_i''$,
  we have $\state[''']\subseteq\state[']$ and $\state[''']\subseteq\state['']$.
  By \cref{result:positive-keep-holding}, we obtain
  $\state[''']\nvDash\varphi_1$
  and
  $\state[''']\nvDash\varphi_2$,
  thus
  $\state[''']\nvDash\varphi_1\vee\varphi_2$.
  Furthermore $\state\sim_i\state[''']$,
  so $\state\nvDash \knows i(\varphi_1\vee\varphi_2)$.

  Further, the ($\Leftarrow$) implication immediately holds by the semantics.  
\end{proof}

\begin{theorem}
  \label{result:ck-disjunction-distributes}
  For any $\varphi_1,\varphi_2\in\mathcal{L}^+$, state \state, and $G\subseteq N$,
\[
\mbox{$\state\vDash \ck G(\varphi_1\vee\varphi_2)$ iff $\state\vDash \ck G\varphi_1\vee \ck G\varphi_2$.}
\]
\end{theorem}
\begin{proof}
The claim follows directly from \cref{result:k-disjunction-distributes} and~\eqref{equ:K}.
\end{proof}

To see that this result does not hold if we allow negation,
consider three players $i,j,k\in N$, $p\in V_k$,
$\state=(\emptyset,\emptyset)$, and
$\varphi=\knows i( \knows j p\vee\neg \knows j p)$.
Then $\state\vDash\varphi$, since the used disjunction is a tautology,
but there is no way for $i$ to know which disjunct is true.

Even with non-tautological disjunctions, the result does not hold.
\begin{example}
  \label{ex:k-disjunction-doesnt-distribute-with-negation}
  With $p\in V_k$ and
  \begin{align*}
    \state&=(\{p\},\{\msg{k}{\{i,k\}}{p}\})\\
    \varphi&=\knows i(\knows j p\vee\neg(\knows j p\vee \knows j\neg p))\enspace,
  \end{align*}
  we have $\state\vDash\varphi$, but again, $i$ knows neither disjunct in $\state$.
  Intuitively, having privately learned that $p$ is true,
  $i$ knows that $j$ either also learned it,
  or that $j$ doesn't know whether $p$ is true,
  but $i$ does not know which of these two statements is true.
\end{example}

Another observation is that
mutual knowledge of any \emph{fact} $p\in\bits$ 
can only be obtained through a corresponding message, 
and is thus inseparably tied to common knowledge.

\begin{lemma}
  \label{result:knowledge-chain-fact-equiv-msg-equiv-ck}
  For any $w\in N^*$ with $\abs{\setof w}\geq2$,
  $p\in\bits$, 
  and state \state,
  the following are equivalent:
  \begin{enumerate}[(i)]
  \item\label{result:knowledge-chain-fact-equiv-msg-equiv-ck:chain}
    $\state\vDash \knows wp$,
  \item\label{result:knowledge-chain-fact-equiv-msg-equiv-ck:msgs}
    there is $\msg{\cdot}{\cdot}{p}\in M_w$,
  \item\label{result:knowledge-chain-fact-equiv-msg-equiv-ck:ck}
    $\state\vDash \ck Gp$ 
    with $G=\setof w$.
  \end{enumerate}
\end{lemma}
\begin{proof}
  $(\ref{result:knowledge-chain-fact-equiv-msg-equiv-ck:chain})\Rightarrow
  (\ref{result:knowledge-chain-fact-equiv-msg-equiv-ck:msgs})$:
  Assume that $(\ref{result:knowledge-chain-fact-equiv-msg-equiv-ck:msgs})$ does not hold.
  Let $V':=V\setminus\{p\}$. 
  Then $(V',M_w)$ is a state and $(V',M_w)\nvDash p$.
  Now let $i\in N$ be such that $p\in\bits_i$.
  Since $\abs{\setof w}\geq2$, there is $j\in\setof w$ with $p\not\in\bits_j$.
  By construction, $(V',M_w)\nvDash p$ and $(V,M)\sim_j(V',M_w)$.
  Since $j\in\setof w$ and both $\state\sim_k\state$ and $(V',M_w)\sim_k(V',M_w)$ for all $k\in N$,
  we obtain $\state\sim_w(V',M_w)$ and thus $\state\nvDash \knows wp$.

  \noindent $(\ref{result:knowledge-chain-fact-equiv-msg-equiv-ck:msgs})\Rightarrow
  (\ref{result:knowledge-chain-fact-equiv-msg-equiv-ck:ck})$:
Suppose that $G=\setof w$ and take $m \in M_w$.
Consider $\state[']$ such that $\state \sim_{G} \state[']$.
This means that for a sequence
$i_1, \LL, i_k$ of players from $G$
  and some states $\state[^1], \LL, \state[^k]$
  we have $\state \sim_{i_1} \state[^1] \sim_{i_2} \LL \sim_{i_k} \state[^k]$,
  where $\state['] = \state[^k]$.
But $i_1 \in G$, so $m \in M_{i_1}$, and consequently $m \in
M^{1}_{i_1}$. Also $i_2 \in G$, so $m \in M^{1}_{i_2}$,
and consequently $m \in M^{2}_{i_2}$.  Continuing this way we conclude
that $m \in M^{k}_{i_k}$, that is $m \in M'_{i_k}$. 

Hence, $m \in M'$. This shows that $M_w \sse M'$.  So $\Facts(M_w) \sse
V'$, since $(V', M')$ is a state. But by the assumption $p \in
Facts(M_w)$, so $\state['] \vDash p$. This proves $\state \vDash \ck G
p$.

  \noindent $(\ref{result:knowledge-chain-fact-equiv-msg-equiv-ck:ck})\Rightarrow
  (\ref{result:knowledge-chain-fact-equiv-msg-equiv-ck:chain})$: By~\eqref{equ:K}.
\end{proof}

We can extend this connection between mutual and common knowledge to arbitrary positive formulas.

\begin{theorem}
  \label{result:knowledge-chain-phi-equiv-ck}
  \label{thm:permutation}
  For any $G\subseteq N$, $\varphi\in\mathcal{L}^+$, and state \state,
  \begin{center}
    $\state\vDash \ck G\varphi$ iff $\state\vDash \knows w\varphi$ for some $w\in G^*$ with $\setof w=G$.
  \end{center}
\end{theorem}
\begin{proof}
  The direction $(\Rightarrow)$ is by~\eqref{equ:K}.

  For $(\Leftarrow)$, we proceed by structural induction.
  The base case
  is obtained from \cref{result:knowledge-chain-fact-equiv-msg-equiv-ck}.
  The induction step for disjunction follows by
  \cref{result:ck-disjunction-distributes},
  and for conjunction it follows directly by definition of the semantics.
  For $\varphi=\knows i\psi$,
  the assumption $\state\vDash \knows w\knows i\psi$ yields, by induction hypothesis,
  that $\state\vDash \ck{G'}\psi$ for $G'=\setof w\cup\{i\}=G\cup\{i\}$,
  which by definition of the semantics implies that $\state\vDash \ck G\knows i\psi$.
\end{proof}
Note that this result provides for positive formulas
a simplified characterization of the common knowledge operator,
as compared with~\eqref{equ:K}.

Finally, we establish a result intuitively saying that a group's
common knowledge of a 
positive formula can only be achieved
when some message (or messages) has been broadcast to at least all
members of this group.  So common knowledge of a positive formula cannot be
achieved among a group by means of more limited communications, for example
point-to-point messages.
Given a formula $\varphi$ we denote by $Facts(\varphi)$ the set of facts that occur in it.

\begin{theorem}
  \label{result:knowledge-chain-phi-only-through-msg}
  \label{thm:pos}
  For any $G\subseteq N$ with $\abs G\geq2$, 
  $\varphi\in\mathcal{L}^+$, and state \state,
  \begin{center}
    if $\state\vDash \ck G\varphi$, then there is $\msg{\cdot}{A}{p}\in M$ with $G\subseteq A$ and $p \in Facts(\varphi)$.
  \end{center}
\end{theorem}
\begin{proof}
  By \cref{result:ck-disjunction-distributes} and the definition of semantics,
  we can transform $\ck G\varphi$ into an equivalent formula
  consisting only of disjunctions and conjunctions over formulas of the form $\ck G\ck{G_1}\dots \ck{G_l}p$
  with $G_1,\dots,G_l\subseteq N$ and $p\in Facts(\varphi)$.
  Since $\state\vDash \ck G\varphi$ 
  there is at least one of these formulas for which
  $\state\models \ck G\ck{G_1}\dots \ck{G_l}p$. 

  Take now $w$ such that $\setof w = G$.
  By~\eqref{equ:K} we obtain $\state\vDash \knows w \ck{G_1}\dots \ck{G_l} p$, so
by the definition of semantics $\state\vDash \knows w p$.
  The claim now follows by \cref{result:knowledge-chain-fact-equiv-msg-equiv-ck}.
\end{proof}

\section{Forwarding}
\label{sec:forwarding}

We now consider a more complex situation in which players are allowed
to send facts that they learned from other players.
We call this contingency `forwarding'.
It is achieved by
relaxing in the definition of a message $(i,A,p)$ the assumption 
$p \in \bits_i$ to $p \in \bits$.
We still insist that a player can send a message only to a group to which he belongs,
that is, that $i \in A$ holds. 

We also assume that only information \emph{known to be true} is sent,
so we now need to examine how a player learned the information he is sending.
This brings us to consider the following relation on the set of messages:
\[
\mbox{$(j,B,p) \leadsto (i,A,p)$ iff $p \not\in \bits_i$ and $i \in B$.}
\]
Intuitively, $(j,B,p) \leadsto (i,A,p)$ means that the fact $p$ is initially not known to player~$i$ 
and that he has learned it from a message sent by player~$j$ to a group to which $i$ belongs.
So $(j,B,p) \leadsto (i,A,p)$ means that $(j,B,p)$ is a possible (partial) \emph{explanation} of $(i,A,p)$.

By a \oldbfe{state} we now mean a pair $(V, M)$ such that for
each message $(i,A,p) \in M$ a sequence of messages $(j_1, B_1, p),
\LL, (j_k, B_k, p)$ exists (i.e., each of these messages about $p$ is in $M$) such
that
\begin{compactitem}
\item these messages form an explanatory chain:
  for $l \in \{1, \LL, k-1\}$ we have $(j_l, B_l, p) \leadsto (j_{l+1}, B_{l+1}, p)$;
\item they are not circular: players $j_1, \LL, j_k$ are all different;
\item $p$ is initially known to player~$j_1$: $p \in V_{j_1}$; and
\item the fact $p$ reaches player~$i$: $(j_k, B_k, p) = (i,A,p)$.
\end{compactitem}
We call such a sequence of messages an \oldbfe{explanation} for $(i,A,p)$ in 
$(V, M)$. So a pair $\state$ is a state if for each of its messages it has an explanation. 

Note that given a state, its messages contain only true facts. That
is, if $(V, M)$ is a state, then $\Facts(M) \sse V$.
Moreover, if $(i,A,p) \in M$, then player $i$ knows that $p$ is true,
i.e.~$(V,M) \vDash \knows i p$.
%
%
%
%
(A more general statement is established in Lemma \ref{lem:equiv_i}.)
Note also that when each message in $M$ is of the form $(i,\cdot,p)$,
where $p \in V_i$, then $(V, M)$ is a state, since each message then forms
its own explanation. So states considered in this section generalize the
states considered in the previous section.

Each state can be alternatively viewed as a partial ordering
$\leadsto^*$ (where $\leadsto^*$ is the reflexive, transitive
closure of the $\leadsto$ relation) on a set of messages such that each message
has an explanation.


The only restriction on the order of the actions comes from the
relation $\leadsto$ that needs to be respected: a player sends a
message that contains information that either he initially knows to be true (the
message is $(i,A,p)$ where $p \in V_i$) or he has learned
(the message is $(i,A,p)$ and some earlier message is of the form
$(j,B,p)$, where $i \in A$).  So the computation begins by some
players who send information they know is true.

We now consider the semantics introduced in \cref{sec:defs} in this
extended setting. It is important to realize that these two semantics
differ in the sense that for a state
and a formula $\varphi \in {\cal L}$ it can happen that
$\state\vDash_H \varphi$ holds in the sense of \cref{sec:defs} but not
in the sense considered now.

\begin{example}
  \label{ex:semantics-different}
  Let $N = \C{i,j,k}$, $H =
  \{\C{i,j}, \C{j,k}\}$, $V = \{p\}$, where $p \in At_i$, and $M =
  \{(i,\{i,j\},p)\}$.  Then we have
  \[
  \state \vDash_H \knows i \neg \knows k p
  \]
  in the sense of \cref{sec:defs}. The intuitive reason is that the fact
  $p$ `belongs' to $i$, so it cannot be used in any message sent by $j$,
  and this information is known to $i$.  However, in the present setting
  $p$ can be used in a message sent by $j$ and we have
  \[
  \state \nvDash_H \knows i \neg \knows k p\enspace.
  \]
  Indeed, consider $(V',M')$ with $V' = \{p\}$ and $M' = \{(i,\{i,j\},p),(j,\{j,k\},p)\}$.
  Then $(V, M) \sim_i (V', M')$ and $(V', M') \vDash \knows k p$.
\end{example}
In general, only non-epistemic formulas have
the same meaning w.r.t.~both semantics.

We now show that some, though not all, properties
established in the previous section also hold in this new setting.
In particular, as in \cref{ex:ck-of-h-matters-for-negative}, we have
$(\ES, \ES) \vDash_H\ck G\neg \knows jp$, so
also now common knowledge can exist without any communication
taking place.
However, as we shall see, \cref{result:ck-of-h-doesnt-matter} does not
hold in general any more.
So in the following we usually consider $\vDash_H$, which, as
mentioned earlier,
includes $\vDash$ as a special case with $H$ being the complete hypergraph.

Recall that for a formula $\varphi$ we denoted the set of facts that occur in it by $Facts(\varphi)$.

\begin{lemma}\label{lem:f}
For any $\varphi \in {\cal L}^+$ and $H$-compliant state $\state$ if $(V,M) \vDash_H \varphi$, then $Facts(\varphi) \cap V \neq \ES$.
\end{lemma}
\begin{proof}
We proceed by structural induction on $\varphi$.
  The only not completely obvious case is when
  $\varphi=\ck G\psi$. But $(V,M) \vDash_H \ck G\psi$ implies $(V,M) \vDash_H \psi$, so in this case the induction hypothesis
readily applies, as well.
\end{proof}

The following result is then a counterpart of \cref{result:knowledge-chain-phi-only-through-msg}.

\begin{theorem}\label{thm:group1}
For any $G \subseteq N$ with $|G| \geq 2$, $\varphi\in\mathcal{L}^+$, and $H$-compliant state $\state$,
\begin{mainclaim}
if $\state\vDash_H \ck G \varphi$, then there is $\msg{\cdot}{A}{p}\in M$ with $G \subseteq A$ and $p \in Facts(\varphi)$.
\end{mainclaim}
\end{theorem}
\begin{proof}
Note that the conclusion of the implication can be written in a more succinct way as
$Facts(\varphi) \cap \Facts(\bigcap_{i \in G} M_i) \neq \ES$.  
Suppose that $\state\vDash_H \ck G \varphi$ and $Facts(\varphi) \cap \Facts(\bigcap_{i \in G} M_i) = \ES$.  
Call a message a \oldbfe{$p$-message} if it is of the form $(\cdot, \cdot, p)$.
Abbreviate $\bigcup_{i \in G} M_i$ to $M_G$
(Note that this is different from $M_w$, where $w$ is a word,
which corresponds to an intersection.)
Three cases arise.
\II

\NI
\emph{Case 1}. For all $p \in Facts(\varphi)$ there is no $p$-message in $M_G$ and $Facts(\varphi) \cap V = \ES$.
Then by \cref{lem:f} $(V,M) \nvDash_H \varphi$.  
\II

\NI
\emph{Case 2}. For all $p \in Facts(\varphi)$ there is no $p$-message in $M_G$ and $Facts(\varphi) \cap V \neq \ES$.

Take some $p \in Facts(\varphi) \cap V$.
Since $|G| \geq 2$, there is $i \in G$ such that $p \not\in V_{i}$.
Remove from $V$ the fact $p$ and from $M$ all $p$-messages.  Denote
the outcome by $(V', M')$. By construction $(V',M')$ is an  $H$-compliant state and by
the assumption there is no $p$-message in $M_{i}$, so
$(V,M) \sim_i (V',M')$. 
\II

\NI
\emph{Case 3}. For some $p \in Facts(\varphi)$ there is a $p$-message in $M_G$.  

Given a set of messages $O \sse
M$, we denote by $top(O)$ the set of $p$-messages $m \in O$, where $p \in Facts(\varphi)$, such that
for no $m' \in O$ we have $m \neq m'$ and $m \leadsto^* m'$.  Further,
we define ($cl$ stands for the closure)
\[
cl(m) := \C{m' \in M \mid m \leadsto^{*} m'}.
\]

We assumed that the set of $p$-messages in $M_G$, where $p \in Facts(\varphi)$, is non-empty,
so the set $top(M_G)$ is non-empty and hence for some $i_1 \in G$ the set
$top(M_G) \cap M_{i_1}$ is non-empty. Choose some $m \in top(M_G) \cap M_{i_1}$.
Let $M' := M \setminus cl(m)$ and $V' := V$.  By the assumption 
$Facts(\varphi) \cap \Facts(\bigcap_{i \in G} M_i) = \ES$, so there is some $i \in G$ such that 
$i \not\in A$, where $m$ is of the form $(\cdot,A,p)$ for some $p \in Facts(\varphi)$.
By the construction $(V', M')$ is an $H$-compliant state
and $(V,M) \sim_{i} (V', M')$.  
\II

We now repeat the above case analysis with $(V',M')$ instead of $(V,M)$.
%
Iterating this way we eventually end up in Case 1 since in Cases 2 and 3
always some fact or message is removed.  This way we obtain a word $w \in
G^*$ and an $H$-compliant state $(V'', M'')$ such that $(V,M) \sim_{w} (V'', M'')$
and $(V'', M'') \nvDash_H \varphi$.  By~\eqref{equ:K} this contradicts the assumption that
$\state\vDash_H \ck G \varphi$.  
\end{proof}

\begin{corollary} \label{cor:iff}
For any $G \subseteq N$ with $|G| \geq 2$, $p\in\bits$, and $H$-compliant state $\state$,
\[
\mbox{$\state\vDash_H \ck G p$ iff there is $\msg{\cdot}{A}{p}\in M$ with $G \subseteq A$.}
\]
\end{corollary}
\begin{proof}
$(\Ra)$ is a direct consequence of \cref{thm:group1}.\\
\NI ($\La$) 
The proof is analogous to the one of the implication
$(\ref{result:knowledge-chain-fact-equiv-msg-equiv-ck:msgs})\Rightarrow
  (\ref{result:knowledge-chain-fact-equiv-msg-equiv-ck:ck})$ of
\cref{result:knowledge-chain-fact-equiv-msg-equiv-ck} and is omitted.
%
\end{proof}

Here is the counterpart of the above result for the case of one
player.  It states that in any state a player knows a fact iff either
he knows it at the outset or he has learned it through a message he
received.

\begin{lemma}
  \label{lem:equiv_i}
For any $i\in N$, $p\in\bits$, and $H$-compliant state $\state$,
\[
\mbox{$\state\vDash_H \knows i p$ iff $p \in V_i \cup\Facts(M_i)$.}
\]
\end{lemma}
\begin{proof}
$(\Ra)$ 
Suppose that $\state\vDash_H \knows i p$ and $p \not\in V_i \cup\Facts(M_i)$.
Remove from $V$ the fact $p$ and from $M$ all messages of the form $(\cdot, \cdot, p)$.
Denote the outcome by $(V',M')$. By construction $(V',M')$ is an $H$-compliant state and by the assumption
$(V,M) \sim_i (V',M')$. So $(V',M') \vDash_H p$, which is a contradiction.
\II

\NI
$(\La)$ 
Consider an  $H$-compliant state $\state[']$ such that $\state \sim_i \state[']$.
Then $V_i = V'_i$ and $M_i = M'_i$. So $V_i \sse V'$ and 
$\Facts(M_i) \sse \Facts(M') \sse V'$, where the final inclusion follows by 
the fact that $\state[']$ is a state. So $p \in V'$ and consequently $\state['] \vDash_H p$, as desired.
\end{proof}

\begin{corollary}
\label{result:ck-of-h-doesnt-matter-for-facts}
For any $i\in N$, $p\in\bits$, and $H$-compliant state $\state$,
\[
\mbox{$\state\vDash_H \knows i p$ iff $\state\vDash \knows i p$.}
\]
\end{corollary}
\begin{proof}
  This follows from \cref{lem:equiv_i} and the fact that $\vDash$ is a special case of $\vDash_H$.
\end{proof}

For further analysis we need an auxiliary concept. Suppose that $(V,M) \sse (V',
M')$, where $(V',M')$ is a state.  In general, $(V,M)$
does not need to be a state but we
can complete it to a state $L(V,M)$ such that $L(V,M)
\sse (V',M')$. Indeed, it suffices for each message $m$ in $M$ to add
to $M$ messages forming an explanation of $m$ in $M'$ and then add
to $V$ the facts used in these added messages.  More precisely, let
$M''$ be a smallest set such that $M \sse M'' \sse M'$ and $(V \cup
Facts(M''), M'')$ is a state. In general, this does not define
a unique state, since each message in $M'$ can have
multiple explanations. However, the states are finite, so we can
always choose $(V \cup Facts(M''), M'')$ in a unique way, for example,
by associating with each state a unique natural number.

From now on we assume that given an inclusion $(V,M) \sse (V', M')$
the state $L(V,M)$ is uniquely defined.  Note that if
\state['] is $H$-compliant, then so are \state and $L\state$.

The following observation will be useful.
\begin{fact}
  \label{fact:L}
For any $i\in N$ and $H$-compliant states $\state ,\state[']$ with $(V_i, M_i) \sse \state[']$,
\[
      \state \sim_i L(V_i, M_i).
\]
\end{fact}
\begin{proof}
All messages in $(V_i, M_i)$ involve player~$i$, so
the $H$-compliant state $\state[''] = L(V_i, M_i)$ is realized by adding to $(V_i,
M_i)$ only some messages that do not involve player~$i$ and some
facts from outside of $\bits_i$.  Consequently $M_i = M''_i$ and
$V_i = V''_i$, that is $\state \sim_i \state['']$.
\end{proof}

Next, the following property of the semantics will be needed.

\begin{lemma}
\label{lem:properties}
For any $H$-compliant state $\state$, $G \sse N$,
  and facts $p_1, \LL, p_k$,
  \[
  \mbox{$\state \vDash_H \ck G(\bigvee_{j=1}^{k} p_j)$ iff $\state \vDash_H \bigvee_{j=1}^{k} \ck G p_j$.}
  \]
\end{lemma}
\begin{proof}
To deal with ($\Rightarrow$) we consider two cases.
\II

\NI
\emph{Case 1}. $|G| = 1$, say $G = \C{i}$.

Suppose that $\state\nvDash_H \bigvee_{j=1}^{k} \knows i p_j$.
Then for $j \in \C{1, \LL, k}$ we have $\state\nvDash_H \knows i p_j$
and thus by
\cref{lem:equiv_i} $p_j \not\in V_i \cup \Facts(M_i)$.
Let
\[
\state[''] := L(V_i, M_i)
\]
be the $H$-compliant state defined w.r.t.~the inclusion $(V_i, M_i) \sse \state$.
This state is realized by adding to $M_i$ some messages from $M$ and to $V_i$ some
facts from $\Facts(M_i)$. So $V'' \sse V_i \cup \Facts(M_i)$ and consequently
for $j \in \C{1, \LL, k}$ we have $p_j \not\in V''$. Hence, 
\[
    \state['']\nvDash_H \textstyle\bigvee_{j=1}^{k} p_j.
\]
Moreover, by \cref{fact:L} we have $\state\sim_i \state['']$, so
$\state\nvDash_H \knows i(\bigvee_{j=1}^{k} p_j)$. 
\II

\NI
\emph{Case 2}. $|G| \geq 2$.

By \cref{lem:f} for some $j \in \C{1, \LL, k}$ there is $\msg{\cdot}{A}{p_j}\in M$ with $G \subseteq A$.
So, by \cref{cor:iff}, $\state\vDash_H \ck G p_j$, and thus $\state \vDash_H \bigvee_{j=1}^{k} \ck G p_j$.
\II

The ($\Leftarrow$) implication holds directly by the definition of the semantics.
\end{proof}

We can now resume our comparison with the results of the previous section. To start with, 
the following result is a counterpart of \cref{result:positive-keep-holding}.

\begin{lemma}
  \label{lem:mono}
  For any $\varphi\in\mathcal{L}^+$ and $H$-compliant states $(V,M)$ and $(V',M')$
  with $(V',M')\subseteq\state $,
\[
    \textup{if } (V', M')\vDash_H\varphi, \textup{ then } \state \vDash_H\varphi.
\]
\end{lemma}

\begin{proof} By structural induction on $\varphi$.
%
%
\end{proof}

In \cref{sec:telling} we used this result to
establish \cref{result:ck-of-h-doesnt-matter}.  However, in the
current setting the counterpart of \cref{result:ck-of-h-doesnt-matter} does not hold.

\begin{example}
  \label{ex:ck-of-h-does-matter}
  Consider players $N = \C{i,j,k,l}$ and a graph $H$
  with the edges $\{l,k\}, \{k,j\}, \{j,i\}$,
  see \cref{fig:line}. Suppose that
  $V = \C{p}$, where $p \in \bits_l$, and $M = \C{(l,\C{l,k},p), (k,\C{k,j},p), (j,\C{j,i},p)}$.
  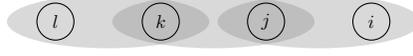
\begin{figure}
    \centering
    \beginpgfgraphicnamed{graph_line}
    \begin{tikzpicture}[scale=.7,transform shape]
      \foreach \x/\n in {2/l,4/k,6/j,8/i}
      \node[draw,commgraphnode] (\n) at (\x,0) {$\n$};
      \begin{pgfonlayer}{background}
      \foreach \n/\nn in {l/k,k/j,j/i}
      \node[commgraphhyperarc,fit=(\n)(\nn)] {};
      \end{pgfonlayer}
    \end{tikzpicture}
    \endpgfgraphicnamed
    \caption{Knowledge of~$H$ matters even for positive formulas when forwarding is allowed.}
    \label{fig:line}
  \end{figure}
  Then
  \[
  \state \nvDash \knows i \knows k p,
  \]
  since player~$i$ does not know through which source player~$j$ learned~$p$. However, 
  \[
  \state \vDash_H \knows i \knows k p,
  \]
  since when the underlying graph is commonly known,
  player~$i$ knows that player~$j$ learned $p$ from player~$k$.
\end{example}

Still, a limited counterpart of \cref{result:ck-of-h-doesnt-matter} does hold.
Let $\mathcal{L}^{+}_{K}$ be the sublanguage of $\mathcal{L}^{+}$ in which
the knowledge operators $\ck G$ are not allowed to be nested. So if $\ck G
\varphi \in \mathcal{L}^{+}_{K}$, then $\varphi$ is a propositional
formula that does not use negation.

\begin{theorem}
  \label{thm:+K}
  For any $H$-compliant state $\state$ and $\varphi\in\mathcal{L}^+_{K}$,
\[
\mbox{$\state\vDash\varphi$ iff $\state\vDash_H\varphi$.}
\]
\end{theorem}
\begin{proof} 
  We proceed by structural induction on $\varphi$.  The only
  non-trivial case is when $\varphi=\ck G\psi$ for some $i\in N$ and
  $\psi$ is a propositional formula that does not use negation.

Let $\bigwedge_{j=1}^{k} \bigvee_{l=1}^{m_j} p_{j,l}$
be the conjunctive normal form of $\psi$. So each $p_{j,l}$ is a fact.
By \cref{lem:properties} and the definition of semantics
we have both
\begin{align*}
  \state\vDash \ck G \psi&\iff\state\vDash\textstyle\bigwedge_{j=1}^{k} \bigvee_{l=1}^{m_j} \ck G p_{j,l}\\
  \intertext{and}
  \state\vDash_H \ck G \psi&\iff\state\vDash_H\textstyle\bigwedge_{j=1}^{k} \bigvee_{l=1}^{m_j} \ck G p_{j,l}\enspace.\\
  \intertext{But by \cref{result:ck-of-h-doesnt-matter-for-facts}, for all $j \in \C{1, \LL, k}$ and $l \in \C{1, \LL, m_j}$ we have}
  \state\vDash \ck G p_{j,l}&\iff\state\vDash_H \ck G p_{j,l}\enspace.
\end{align*}
This implies the claim for $\ck G\psi$.
\end{proof}

We now analyze to what extent \cref{result:ck-disjunction-distributes} holds in the current setting.
We first prove that the $\ck G$ operator distributes over
disjunctions of formulas from the non-epistemic sublanguage
$\mathcal{L}_{\wedge,\vee}$ of $\mathcal{L}$ in which only
conjunction and disjunction is allowed.

\begin{theorem}
  \label{thm:disj-p}
  For any $\varphi_1,\varphi_2\in\mathcal{L}_{\wedge,\vee}$, $i\in N$ and $H$-compliant state $\state$,
  \[
  \mbox{$\state \vDash_H \ck G(\varphi_1\vee\varphi_2)$ iff $\state \vDash_H \ck G\varphi_1\vee \ck G\varphi_2$.}
  \]
\end{theorem}
\begin{proof}
Passing by the conjunctive normal forms of $\varphi_1$ and $\varphi_2$,
the result follows from the definition of the semantics and Lemma \ref{lem:properties} twice.
\end{proof}



However, the $\ck G$ operator does not distribute over the knowledge operators,
so the counterpart of \cref{result:ck-disjunction-distributes} does not hold.

\begin{example}
  \label{ex:ck-doesnt-distribute-over-k}
  Consider the set of players $N = \C{i,j,k,l,n}$
  and the hypergraph $H$ being the graph with the edges
  $(n,k), (n,l), (k,j), (l,j), (j,i)$,
  see \cref{fig:interaction-structure-examples}\subref{fig:interaction-structure-examples:b}.
  Take $V = \C{p}$, where $p \in \bits_n$, and
  \[
  M =
  \C{(n,\C{n,k},p), (k,\C{k,j},p), (j,\C{j,i},p)}\enspace.
  \]
  Then
  \[
  \state \vDash_H \knows i(\knows k p \vee \knows l p),
  \]
  but neither
  $
  \state \vDash_H \knows i \knows k p,
  $
  nor
  $
  \state \vDash_H \knows i \knows k q
  $
  holds. Informally, player~$i$ knows that either player~$k$ or player~$l$ knows~$p$
  but he does not know which one of them knows~$p$.
\end{example}

As noticed already after the proof of \cref{result:ck-disjunction-distributes},
the $\knows i$ operator does not distribute over negation either;
the same example applies here.

Finally, reconsider \cref{thm:permutation}. It is straightforward to see that it does not hold in 
the present setting, even for two players. Indeed, reconsider \cref{ex:ck-of-h-does-matter}.
We showed there that $\state \vDash_H \knows i \knows k p$. However, it is easy to see that 
$\state \nvDash_H \ck{\{i, k\}} p$ since $\state \nvDash_H \knows k \knows i p$. 

\section{Conclusions and related work}
\label{sec:conclusions}

In this paper we studied various aspects of common knowledge in two
simple frameworks concerned with synchronous communication. It is
useful to clarify that our two impossibility results concerning
the attainment of common knowledge amongst players
(\cref{result:knowledge-chain-phi-only-through-msg,thm:group1})
differ from the customary impossibility results.

For example, \citet{HM90} formalize the epistemic aspects of the
celebrated Coordinated Attack Problem that consists in
achieving common knowledge (a `common plan of action').
They show (in Section~8) that in a
distributed system in which communication is not guaranteed, common
knowledge is not attainable. When communication is guaranteed, they
show the same result when there is no bound on message delivery
times. In both situations the proof assumes the existence of clocks and
point-to-point communication.

The close correspondence between simultaneous events (in our system a
broadcast to the whole group) and common knowledge is pointed out by
\citet{FHMV99}. Their model of a distributed system consists of a
set of linear `runs' (histories), while we only assume a partial
ordering ($\leadsto$) between messages broadcast to groups, which are
the only possible actions.
We have shown that in our framework, common knowledge of a positive formula
is indeed inseparably related to group communication, which corresponds to simultaneous events.
However, as we have seen, this does not hold of negative formulas,
so the relationship is not as obvious as it may seem.
The results of \citet{FHMV99} may be seen to correspond to our \cref{cor:iff},
though we allow broadcasts instead of just point-to-point communication.

\citet{chandy_processes_1986} consider the flow of information
in distributed systems with asynchronous communication.
They study how processes `learn' about states of other processes and how knowledge evolves.
The main difference is that with asynchronous communication,
hypergraphs are equivalent to mere point-to-point graphs.
Without guarantees on the delivery time, and without temporal reasoning,
from the knowledge point of view
sending an asynchronous group message has the same effect as
sending a separate message to each group member.

Our study concerning the consequences of the assumption whether the underlying
hypergraph is commonly known among the players brings our paper
somewhat closer to the area of social networks
(see, e.g., \citet{Jac08}).
Within logic, the relevance of epistemic issues in communication networks
has been recognized by a number of authors, e.g.~\citet{van_benthem_one_2006_}.
However, to our knowledge the only work that addresses these issues
is \citet{pacuit_reasoning_2007} and, to some extent,~\citet{roelofsen_exploring_2005}.
We now briefly discuss these frameworks and relate them to our own.

\Citet{pacuit_reasoning_2007} use a history-based model
to study diffusion of information in a communication graph,
starting from facts initially known to individual players.
Communicative acts are assumed to consist in
a player~$j$ `reading' an arbitrary propositional formula from another player~$i$,
with the precondition that~$i$ \emph{knows} that the formula holds.
Communicative acts are restricted to a commonly known, static, directed graph,
and, unlike in our case, are assumed to go \emph{unnoticed by~$i$}.
The paper formalizes what conclusions,
beyond the mere factual content of messages,
can be drawn using knowledge of the communication graph and, consequently,
knowledge of the possible routes along which certain information can have flown.

\Citet{roelofsen_exploring_2005} uses a model based on Dynamic
Epistemic Logic (DEL) to describe how some initial epistemic
model evolves in a communication situation.  Communication is
among subgroups and can contain arbitrary epistemic formulas.
Further, communication is assumed to be truthful and is restricted to
occur along a hypergraph.
However the hypergraph is explicitly encoded in the model, and thus
(knowledge of it) is subject to change.

While under certain circumstances history-based modeling and DEL are
equivalent~\cite{van_benthem_merging_2007}, our approach is more
in the spirit of~\citet{pacuit_reasoning_2007}.
Indeed, we also study how specific information may have spread.
Also, all possible communications are included in the model
and suspicions about them are not explicitly formed.
Finally, the underlying graph (in our case hypergraph) is static
and not included in the model.

On a technical level, our approach differs from~\citet{pacuit_reasoning_2007}
in that we use sets of messages instead of sequences
and, when dealing with forwarding, employ a more general structure than histories
by considering messages partially ordered by the relation $\leadsto$.
On the other hand, our messages are simpler:
\citet{pacuit_reasoning_2007} allow disjunctions of facts, while we allow only facts.

What distinguishes our approach on a more conceptual level
is that our focus lies on identifying natural conditions that allow us to
prove stronger results about knowledge, such as distributivity over
disjunctions, or irrelevance of (common) knowledge of the underlying
hypergraph.


\section{Extensions}
\label{sec:extensions}

We conclude by listing a number of natural extensions of the considered framework 
that are worthy of further study:
\begin{itemize}

\item We could equip the players with theories that their parts of
  valuations, $V_i$, have to satisfy.  In this extension we would
  assume that each player $i$ has a propositional theory $T_i$
  built from facts in $At_i$ that he adheres to. The theories $T_i$
  where $i \in N$ then form a common knowledge among the players. So
  each player $j$ can assume that player $i$ considers $V_i$ such that
  $V_i$ is a model of $T_i$.

\item We could consider more complex messages than simple atomic
facts, for example propositional formulas, or even
epistemic formulas.
Also, we could study asynchronous communication,
messages from unknown senders or to an unknown group of recipients,
and a counterpart of the blind copy feature familiar from e-mails.
  
\item In \cref{sec:forwarding} we relaxed the assumption that in a
  message $(i,A,p)$ it has to be the case that $p \in \bits_i$, but we
  did still insist on the \emph{truthfulness} of messages, requiring
  that $p \in V$.  We could further relax this assumption, by
  insisting only that $p \in At$. This way we would model messages
  that consist of possibly false (but credible) information.  This
  would lead to a study of beliefs (which can be false) rather than
  knowledge (which cannot) and common beliefs rather than common
  knowledge.
  
\item We could consider in this framework belief revision, by assuming
  that the theory $T_i$ of player $i$ consists of his beliefs, which
  would then be revised in view of received information.
  Alternatively, $T_i$ could be the certain knowledge of player $i$
  against which received information would be revised.
  
\item We could assume that the players have different knowledge of the
  underlying hypergraph, by assuming that for all $i$ we have $H \sse
  H_i$, where $H$ is the underlying hypergraph and $H_i$ is its
  approximation known to $i$, and that players learn $H$ by exchanging
  messages.  The messages would contain information about which
  hyperarcs do \emph{not} belong to $H$.
\item Alternatively, we could study a setup in which each player has
an indistinguishability relation over hypergraphs. This
would allow us to model players' partial knowledge of the
underlying hypergraph.
\end{itemize}

We use the setting of the first item in~\cite{apt_strategy_2009} to reason
about iterated elimination of strategies in \oldbfe{strategic games with interaction structures}.
These are strategic games in which there is a hypergraph over the set
of players (an interaction structure) and
the players can communicate about their preferences, initially only known to themselves,
so that within each hyperarc players can obtain common knowledge of each other's
preferences.

\section*{Acknowledgements}
\label{sec:acknowledgements}

We thank Rohit Parikh, Willemien Kets, Aaron Ar\-cher, Henry Landau,
and three anonymous referees for discussion and helpful suggestions.
The second and third authors were supported by a GLoRiClass fellowship
funded by the European Commission (Early Stage Research Training
Mono-Host Fellowship MEST-CT-2005-020841).









\end{document}